\documentclass{amsart}
\usepackage{graphicx,enumerate,amssymb}
\usepackage{tikz}
\usepackage{xcolor}
\usepackage{hyperref}
\newcommand{\CC}{\mathbb C}
\newcommand{\RR}{\mathbb R}
\newcommand{\NN}{\mathbb N}
\newcommand{\FF}{\mathcal F}

\newcommand{\tsp}{^{\mbox{\scriptsize\sc t}}}
\newcommand{\ket}[1]{|#1\rangle}

\newtheorem{definition}{Definition}
\newtheorem{notation}[definition]{Notation}
\newtheorem{example}[definition]{Example}
\newtheorem{lemma}[definition]{Lemma}
\newtheorem{theorem}[definition]{Theorem}
\newtheorem{proposition}[definition]{Proposition}
\newtheorem{remark}[definition]{Remark}
\begin{document}
\title{Quantum entanglement and the Bell Matrix}
\thanks{PACS number: 03.67.Bg, 03.67.Mn}
\thanks{2010 Mathematics Subject Classification: 81P68, 81P40 }
\thanks{Partially supported by PRIN2011 Project ``Metodi Logici per il 
trattamento dell'informazione''}
\author[Lai, A. C.]{Anna Chiara Lai}
\address{Dipartimento di Matematica e Fisica,
          Universit\`a Roma Tre,
          Via della Vasca Navale 84,
          00146 Roma, Italy}
\email{aclai@mat.uniroma3.it}
\author[Pedicini, M.]{Marco Pedicini}
\address{Dipartimento di Matematica e Fisica,
          Universit\`a Roma Tre,
          Via della Vasca Navale 84,
          00146 Roma, Italy}      
\email{marco.pedicini@uniroma3.it}
\author[Rognone, S.]{Silvia Rognone}
\address{Dipartimento di Matematica e Fisica,
          Universit\`a Roma Tre,
          Via della Vasca Navale 84,
          00146 Roma, Italy}
\email{}

\keywords{Quantum entanglement, Entangled states, Multipartite entanglement, 
CNOT, Bell states}
\maketitle
\begin{abstract}
We present a class of maximally entangled states generated by a
high-dimensional generalisation of the \textsc{cnot} gate. The
advantage of our constructive approach is the simple algebraic
structure of both entangling operator and resulting entangled
states. In order to show that the method can be applied to any
dimension, we introduce new sufficient conditions for global and
maximal entanglement with respect to Meyer and Wallach's measure.
\end{abstract}

\section{Introduction}

Entanglement is a key feature of quantum mechanical systems with wide
applications to the field of quantum information theory.  The class of
quantum processes relying on entangled states include quantum state
teleportation \cite{teleportation}, quantum error correction
\cite{error}, quantum cryptography \cite{key}, and some quantum
computational speedups \cite{speedup}.  Multi-qubit entangled states
are regarded as a valuable resource for processing information: for
instance, several authors applied multi-qubit entanglement (and
related entangling procedures) to multi-agent generalizations of the
quantum teleportation protocol introduced in the paper by Bennett,
Brassard, Cr\`epeau, Jozsa, Peres, and Wootters \cite{teleportation}
-- see for instance \cite{tele1}.  Also, other classes of multi-qubit
entangled states turned out to be suitable for superdense coding.

Applications to quantum information theory motivated the search for
the mathematical characterisations of multi-particle entanglement and
for highly entangled quantum states. The approaches to this problem
include an analytical classifications of entangled states
\cite{popescu, vidal}, numerical optimisation techniques \cite{num1},
and geometric characterisations \cite{top}.

Here we present a class of maximally entangled states, that we call
\emph{general Bell states} or \emph{$2^n$-dimensional Bell states},
generated by an arbitrarily high-dimensional generalisation of the
\textsc{cnot} gate. The advantage of our approach is the simple
algebraic structure of both entangling gates and resulting states. In
order to show the full generality of the method, we prove new
sufficient conditions for both global entanglement and maximal
entanglement (with respect to Meyer and Wallach's measure, see
Equation~\eqref{meyer}): being based on the expectation value of an
explicitly given operator, these criteria feature a simple
formulation, scalability and observability.

In \cite{filters} Osterloh and Siewert propose a general method to
construct new classes of entanglement measures based on suitable
products and combinations of Pauli's matrices.  Inspired by this
approach, as well as by the multi-qubit concurrence proposed in
\cite{braid} and by the relation between antilinear operators and
concurrence \cite{concurrencefilters}, in what follows, we introduce a
particular antilinear operator (Definition~\ref{defop}) and we use its
expectation value as an entanglement criterion
(Proposition~\ref{pcrit}) for general Bell states.  In
Proposition~\ref{pmax}, we show that such an operator turns out to be
related to Meyer and Wallach's (MW) measure \cite{meyerwallach} and we
employ this relation to show that the general Bell states are
maximally entangled with respect to this measure -- Theorem~\ref{thm}.

To the best of our knowledge, an univoque and commonly accepted notion
of entanglement measure in high-dimensional systems has not yet been
introduced.  Several proposals in the literature try to capture
distinct aspects of a maximally entangled state. For instance, the
\emph{Schmidt decomposition}, see \cite{schmidtmeasure}, induces a
measure related to the minimum number of terms in the product
expansion of a state, while the \emph{fully entangled fraction}
measures the ability of a state to perform tasks related to quantum
computing, such as teleportation and dense coding \cite{fraction}.

Throughout this paper, we focus on MW measure
\cite{meyerwallach}. This measure interprets the global entanglement
as the average bipartite entanglement of every qubit with respect to
the rest of the system.  It has thus the advantage of a simple
physical meaning as well as a simple formulation, introduced in
\cite{brennen}:
\begin{equation}\label{meyer}
Q(|\psi\rangle):=2\left(1-\frac{1}{n}\sum_{j=1}^{n} Tr[\rho_j^2]\right) 
\end{equation}
where $n$ is the number of qubits of the system, $\rho_{j,\psi}$ is
the density matrix obtained by tracing out the $j$-th qubit of the
state $|\psi\rangle$ and $Tr[\cdot]$ represents the trace operator.

The main result we present here is a sufficient condition on
multi-qubit states to be maximally entangled (with respect to MW
measure) and, as mentioned above, we establish this result in order to
show that a set of states generalising Bell states have maximal MW
measure.  \vskip0.5cm

The paper is organised as follows. In Section \ref{s2} we show
sufficient conditions for global entanglement and for the maximality
of the MW measure of a state in a multi-qubit system.  In Section
\ref{s3} we propose a generalisation of the \textsc{cnot} gate to
multi-qubit systems a related class of states, that we call
$2^n$-dimensional Bell states.  By applying the criteria introduced in
Section \ref{s2}, we are able to show that these generalisations of
Bell states are maximally entangled with respect to MW measure. Some
possible extensions of this approach are illustrated in Section
\ref{s31}.

\section{An entanglement criterion}\label{s2}
First of all we give the formal definition of \emph{globally
  entangled} state.
\begin{definition}
A state $|\psi\rangle$ is \emph{globally entangled} if for any
$|\phi_1\rangle$ and $|\phi_2\rangle$ we have $|\psi\rangle \neq
|\phi_1\rangle \otimes |\phi_2\rangle$.
\end{definition}

\begin{remark}
Throughout this paper we consider elements of Hilbert spaces
$|\psi\rangle\in\CC^{2^n}$ which are \emph{pure quantum states}, i.e.,
they are complex vectors of unit Euclidean norm:
$|\psi\rangle=(\psi_1,\dots,\psi_{2^n})$ and $\sum_{j=1}^{2^n}
|\psi_j|^2 = 1$; for brevity we refer to them simply as ``states''.
\end{remark}

\begin{notation}
We use the symbol $I_{2^n}$ to denote the $2^n$-dimensional identity
matrix:
$$ I_{2^n} := \underbrace{I_2 \otimes \dots \otimes I_2}_{n-\text{times}}.$$
being $I_{2^n} = (1)$ if $n=0$.
 
The \emph{expectation value} of the operator $A$ in the state $\psi$
is denoted by
$$\langle A \rangle_\psi := \langle \psi|A| \psi\rangle.$$  

Moreover we denote by $\sigma_y$ the Pauli matrix 
$$\begin{pmatrix}
  0&-i\\ 
  i&0
\end{pmatrix}.$$
\end{notation}

We introduce the following two operators, they are used to define the
particular antilinear operator we apply to states constructed with
algorithm in Section~{\ref{s3}} in order to prove they are entangled
states.
\begin{definition}\label{defop}
Let us denote by $\FF : \CC^{2^n} \to \CC$ the function which
associates to a state $|\psi\rangle$ the expectation value of the
operator $M_{2^n}K_{2^n}$ in the state $|\psi\rangle$, namely:
\begin{equation}
\FF(|\psi\rangle):= 
\langle M_{2^n}K_{2^n}\rangle_\psi 
\end{equation}
where 
$M_{2^n}:=\sigma_y\otimes I_{2^{n-2}}\otimes\sigma_y$ and
$K_{2^n}$ is the \emph{conjugation operator}.             
\end{definition}

\begin{figure}[top!]
\begin{center}
\begin{tabular}{ccc}
\includegraphics[width=0.3\linewidth]{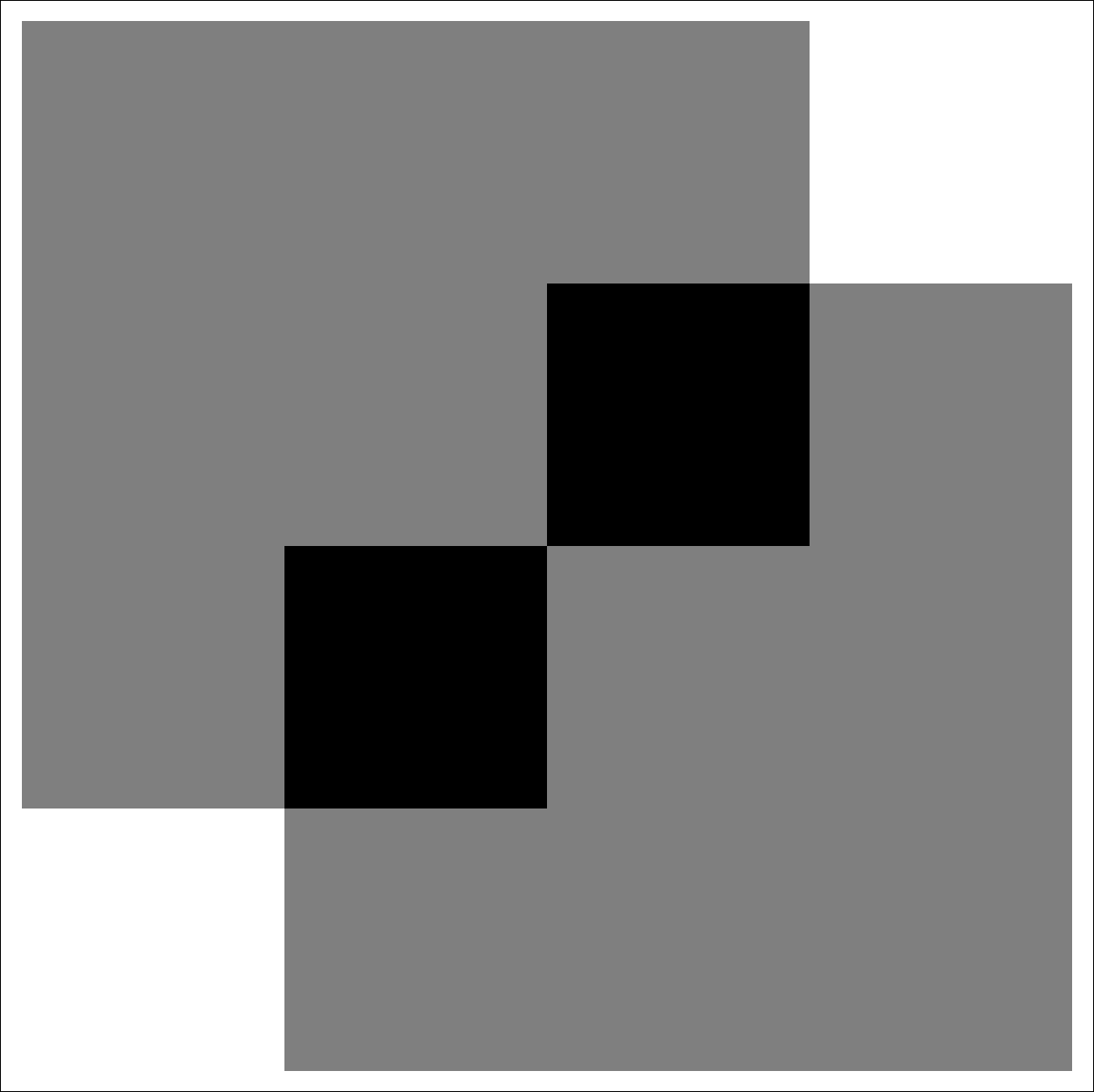} &
\includegraphics[width=0.3\linewidth]{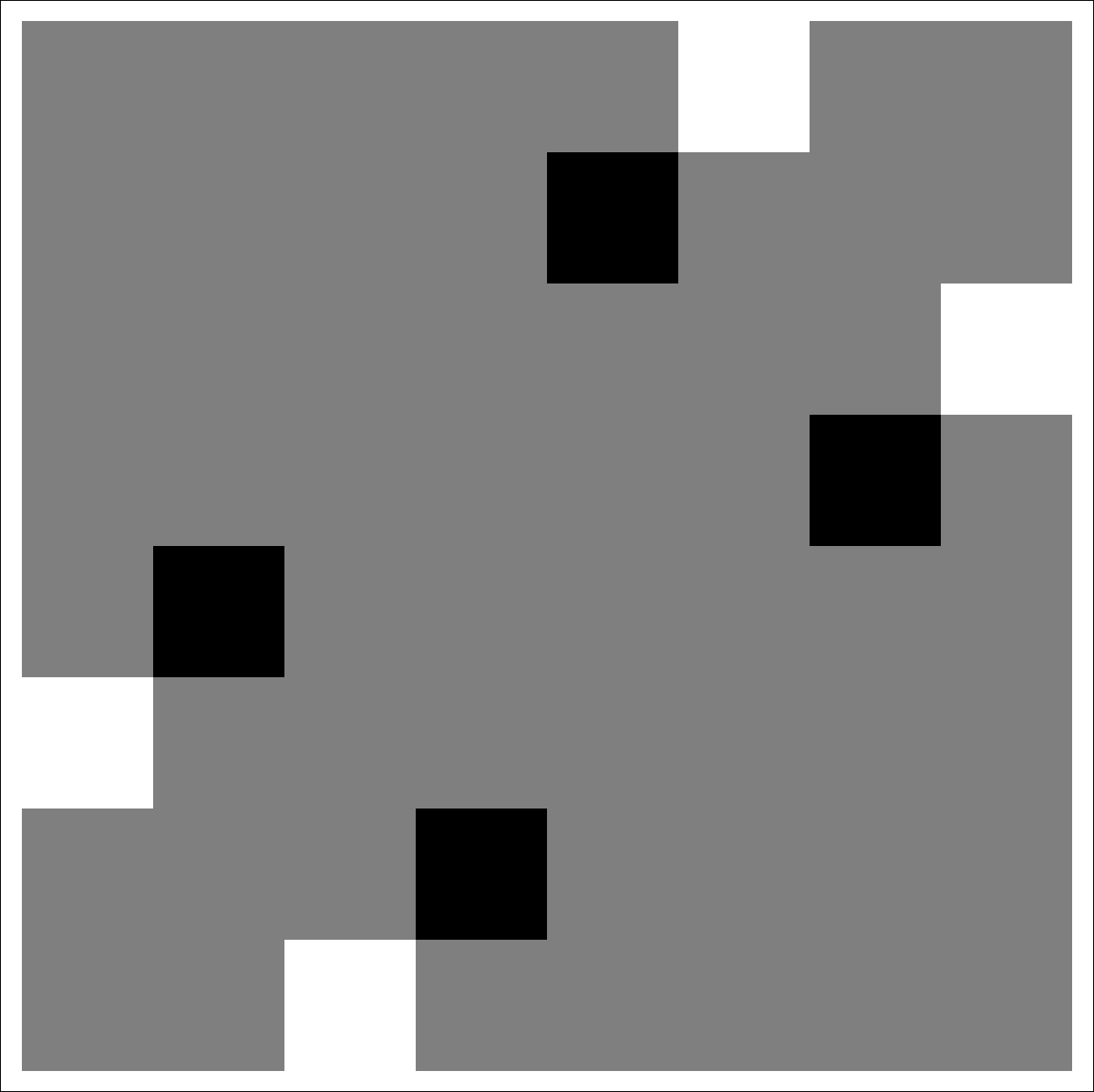} &
\includegraphics[width=0.3\linewidth]{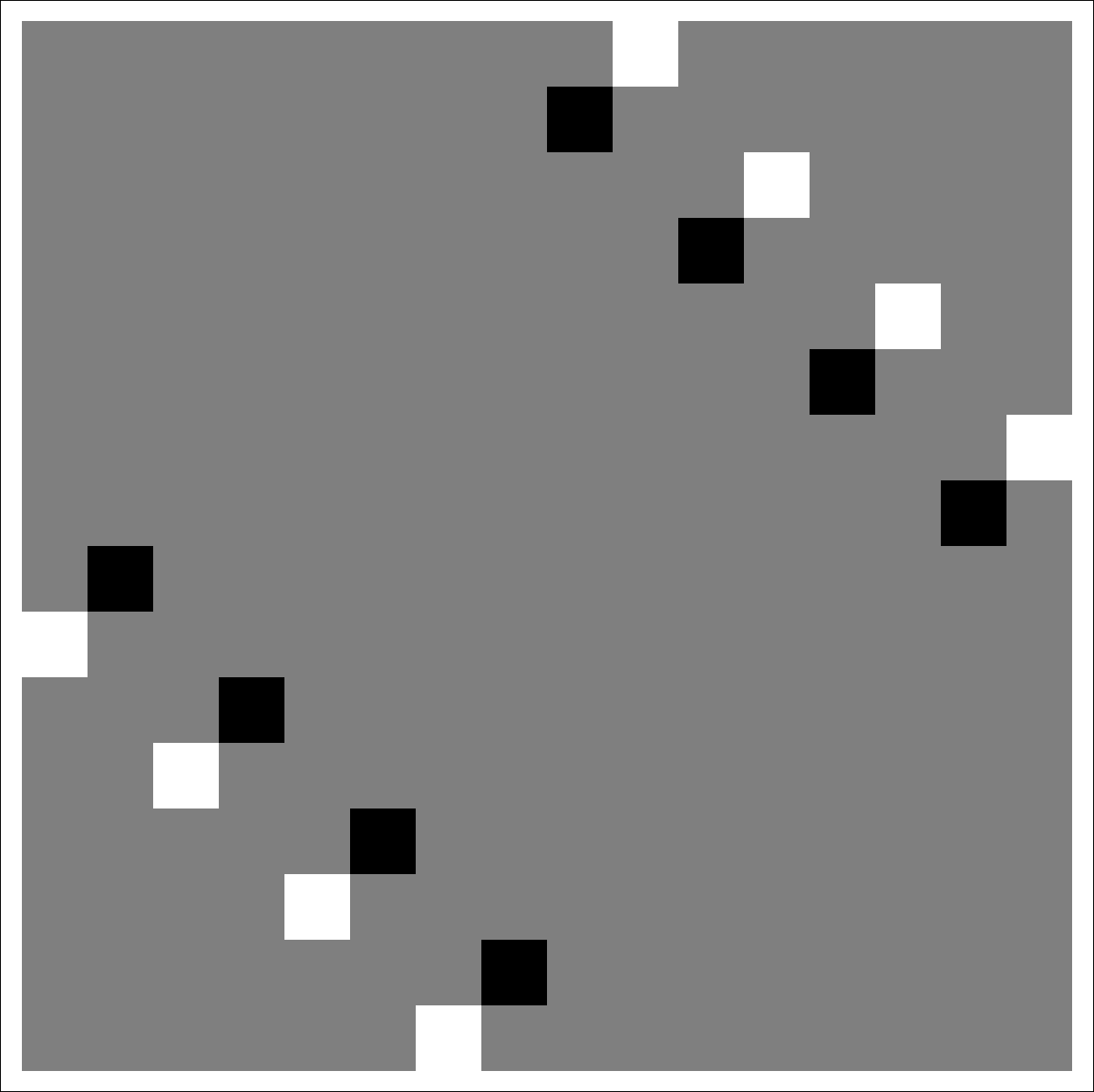} \\
\includegraphics[width=0.3\linewidth]{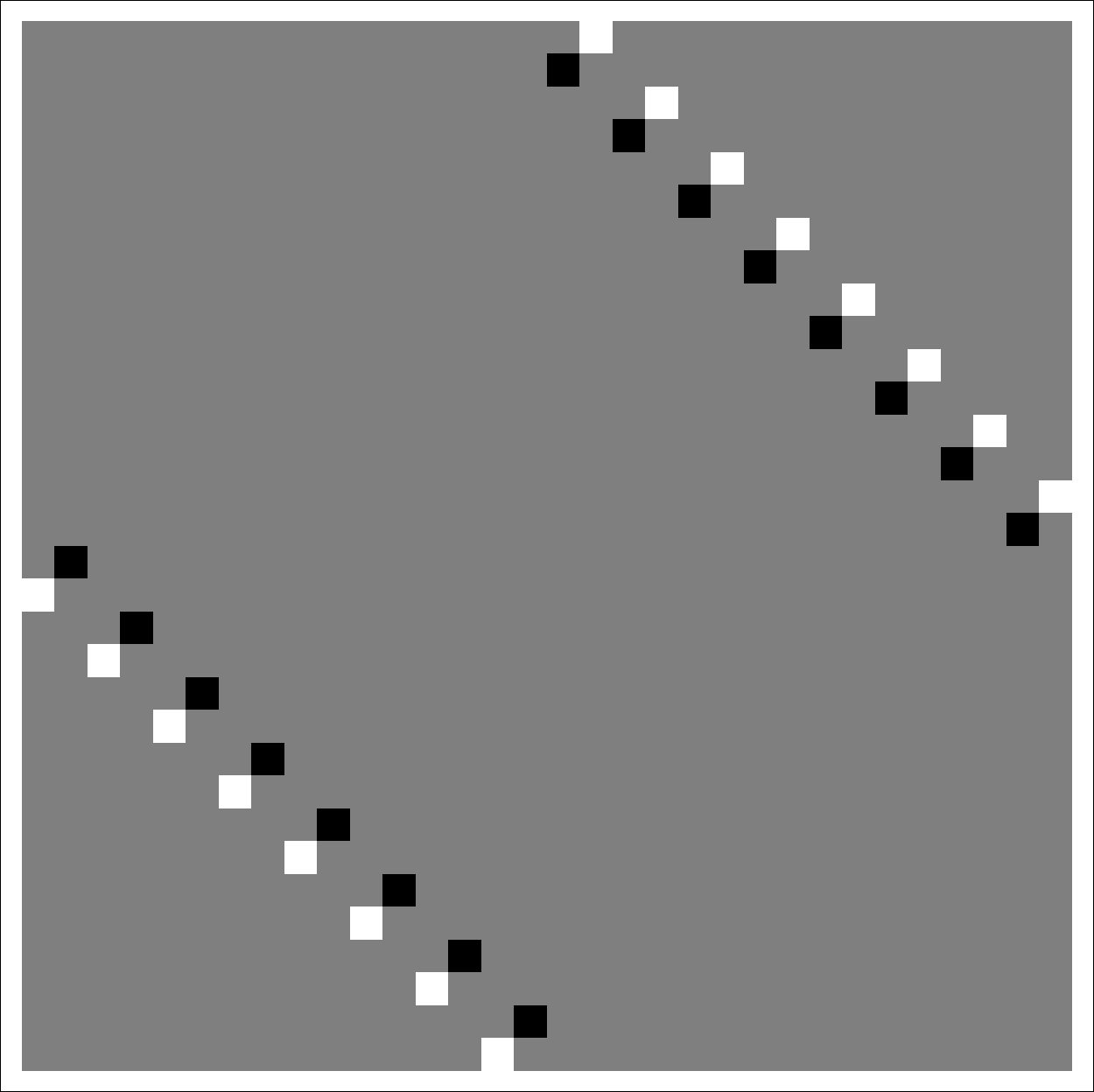} &
\includegraphics[width=0.3\linewidth]{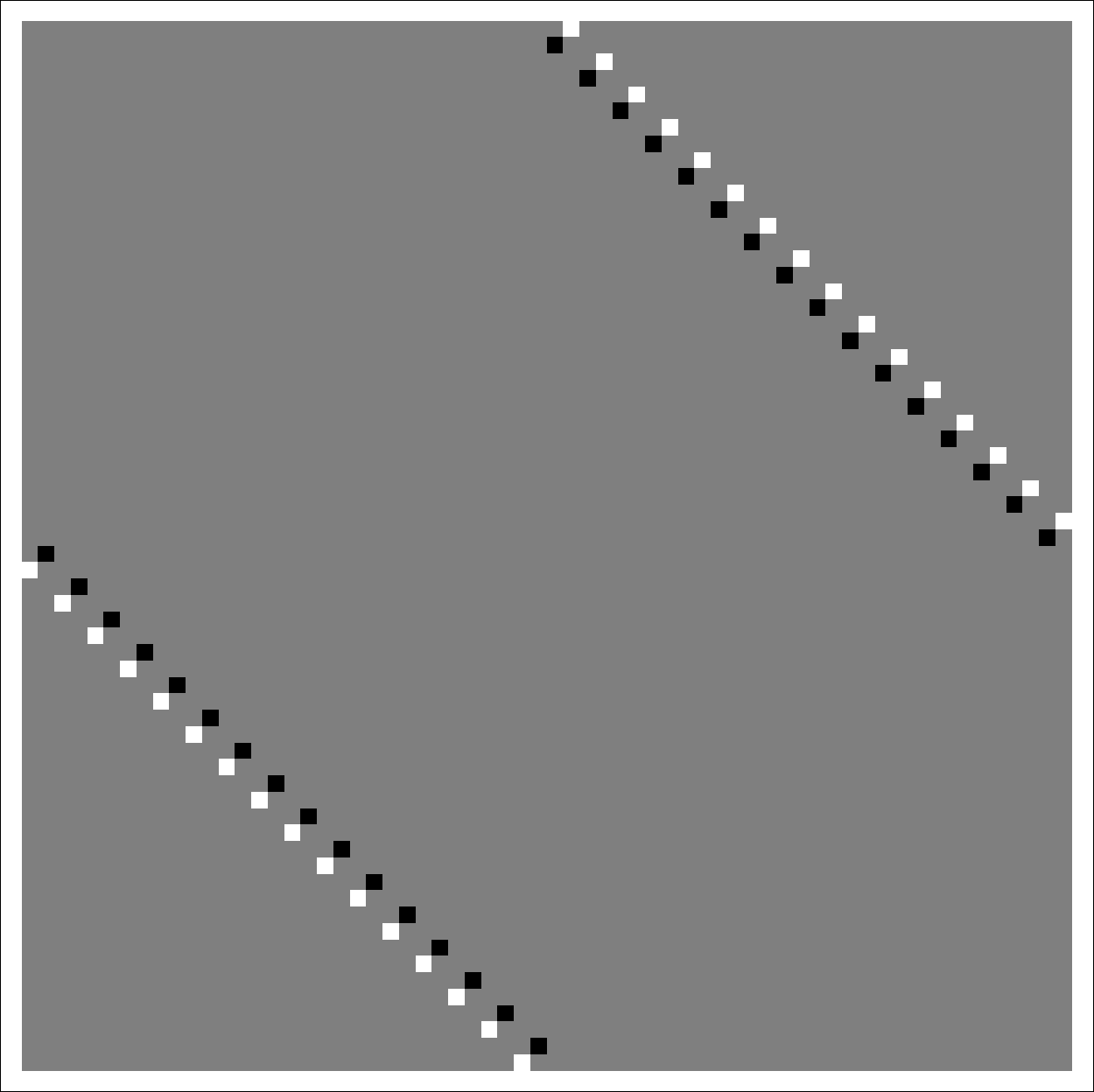} &
\includegraphics[width=0.3\linewidth]{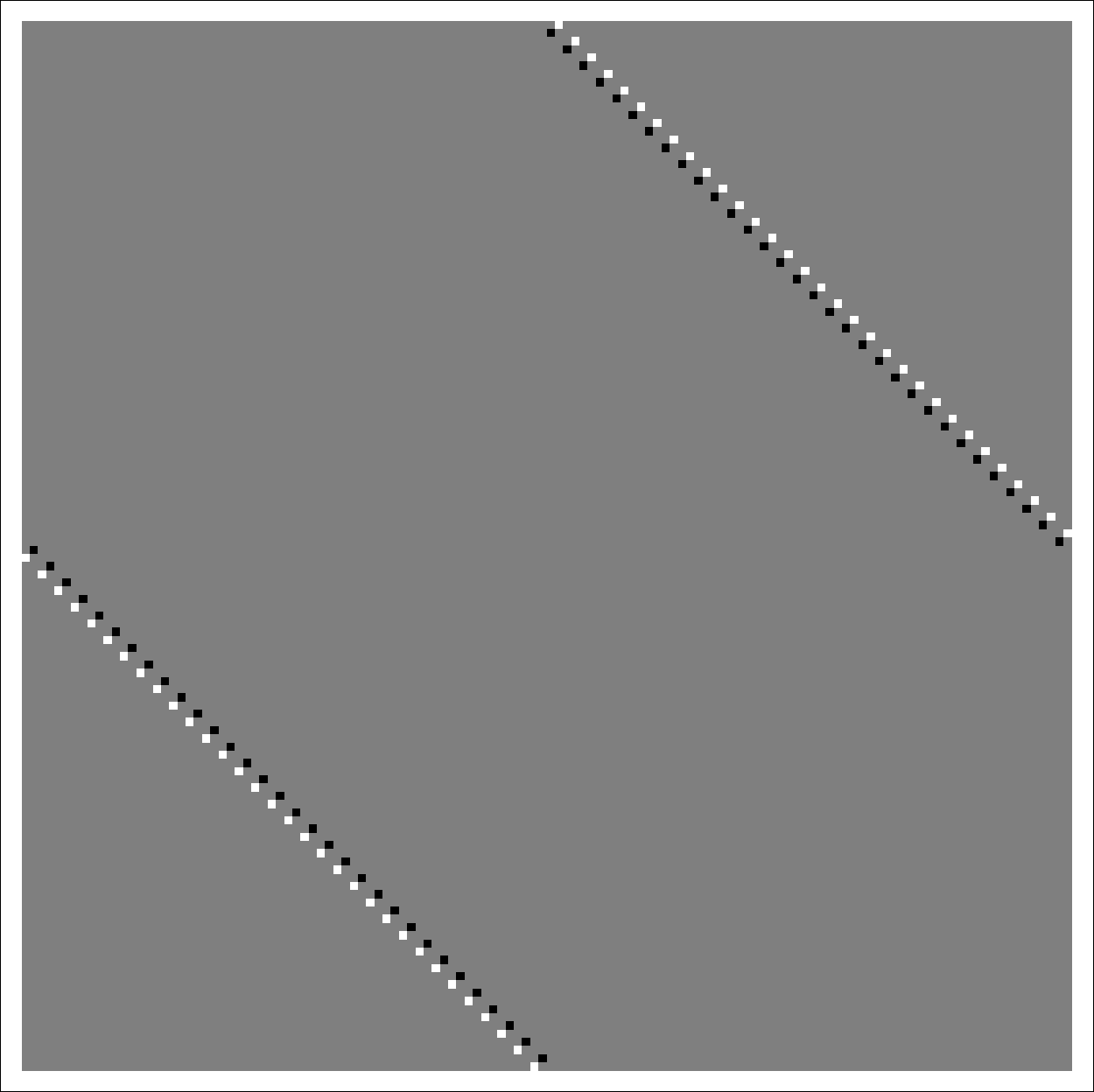} 
\end{tabular}
\end{center}
\caption{We show in this picture the matrices $M_{2^n}$ for $n=2,\dots
  ,7$.  Entries $0$ are shown in grey color, entries $+1$ by black
  color and entries $-1$ by white color.}\label{figM}
\end{figure}
Note that $$\FF(|\psi\rangle):= \langle
M_{2^n}K_{2^n}\rangle_\psi=\langle \psi|M_{2^n}K_{2^n}| \psi\rangle =
\langle \psi|M_{2^n}|\bar \psi\rangle$$ where $|\bar \psi\rangle$
denotes the complex conjugate of $|\psi \rangle$.
\medskip

\begin{example}
We explicitly compute $M_4$:
\begin{align*}
M_4 &:= \sigma_y\otimes I_{2^{2-2}}\otimes\sigma_y = 
\sigma_y\otimes(1)\otimes\sigma_y =\\
	& = \sigma_y\otimes\sigma_y = \begin{pmatrix}
0&-i\\
i&0
\end{pmatrix} \otimes\sigma_y= \begin{pmatrix}
0&-i \sigma_y\\
i \sigma_y&0
\end{pmatrix} 
 = \begin{pmatrix}
0 & 0 & 0 & -1 \\
0 & 0 & 1 & 0 \\
0 & 1 & 0 & 0 \\
-1 & 0 & 0 & 0 \\
\end{pmatrix} 
\end{align*}
For a representation of matrices $M_{2^n}$ with $n\geq 2$ see Figure
\ref{figM}.
\end{example}

We now show that $M_{2^n}K_{2^n}$ has zero expectation value on
product states.
\begin{proposition}\label{pcrit}
If $|\psi\rangle$ is an unentangled state then $\FF(|\psi\rangle)=0$.
\end{proposition}
\begin{proof}
Let $n\geq 1$, $|\psi\rangle\in\CC^{2^n}$, $|\phi_1\rangle\in
\CC^{2^{n_1}}$, $|\phi_2\rangle\in \CC^{2^{n_2}}$, $n_1+n_2=n$ and
assume $|\psi\rangle=|\phi_1\rangle \otimes |\phi_2\rangle$.  Also let
$|\bar \phi_1\rangle = (a,b) \in \CC^{2^{n_1}}$ defined by $a,b\in
\CC^{2^{n_1-1}}$, two half vectors of its coordinates in the standard
basis.  One has
$$(\sigma_y\otimes I_{2^{n_1-1}}) (a,b)=-i(b,-a).$$
Consequently
\begin{equation}\label{e01}
\langle \phi_1|\sigma_y\otimes I_{2^{n_1-1}} |\bar \phi_1\rangle= 
(a,b)\tsp(\sigma_y\otimes I_{2^{n_1-1}}) (a,b)=-i a\cdot b+ i b\cdot a=0.
\end{equation}
Similarly, if $|\bar \phi_2\rangle=(y_1,\dots,y_{2^{n_2}})$ then
\begin{align*}
( I_{2^{n_2-1}}\otimes \sigma_y) |\bar \phi_2\rangle
&=( I_{2^{n_2-1}}\otimes  \sigma_y) (y_1,\dots,y_{2^n})\\ 
&=-i(y_2,-y_1,y_3,-y_4,\dots,y_{2^{n_2}},-y_{2^{n_2}-1})
\end{align*}
thus
\begin{equation}\label{e02}
\begin{split}
\langle \phi_2| I_{2^{n_2-1}}& \otimes \sigma_y|\bar 
\phi_2\rangle=(y_1,\dots,y_{2^{n_2}})\tsp( I_{2^{n_2-1}}\otimes \sigma_y) 
(y_1,\dots,y_{2^{n_2}})\\
&=-i(y_1,\dots,y_{2^{n_2}})\tsp(y_2,-y_1,y_3,-y_4,\dots,y_{2^{n_2}},-y_{2^{n_2}
-1})=0.                            
\end{split}
\end{equation}
By equalities (\ref{e01}) and (\ref{e02}) one finally has 
\begin{align*}
\FF(\ket\psi)&=  \langle \psi|M_{2^n}|\bar \psi\rangle=\langle \phi_1 \otimes 
\phi_2|\sigma_y\otimes I_{2^{n-2}}\otimes\sigma_y|\bar \phi_1 \otimes \bar 
\phi_2\rangle\\
&=\langle \phi_1 \otimes \phi_2|\sigma_y\otimes I_{2^{n_1-1}}\otimes  
I_{2^{n_2-1}}\otimes\sigma_y|\bar \phi_1 \otimes \bar \phi_2\rangle\\
&=\langle \phi_1|\sigma_y\otimes I_{2^{n_1-1}} |\bar \phi_1\rangle 
\langle\phi_2|I_{2^{n_2-1}}\otimes\sigma_y|\bar \phi_2\rangle=0.
\end{align*}
\end{proof}

Next result shows that $\FF$ also provides a sufficient condition for
maximal entanglement.  It is useful to recall the following
\begin{definition}[Schimdt decomposition] Let $n_1,n_2\in\NN$ such that 
$n_1+n_2=n$ and let $A=\CC^{2^{n_1}}$ and $B=\CC^{2^{n_2}}$ so that 
$\CC^{2^n}=A\otimes B$. Then any state $|\psi\rangle\in \CC^{2^n}$ can be 
written in the form 
$$|\psi\rangle = \sum_{k=1}^K c_k|\phi_k^A\rangle\otimes
|\phi_k^B\rangle$$ where $K=\min\{dim(A),dim
(B)\}=\min\{2^{n_1},2^{n_2}\}$, $c_k\geq 0$ and
$\{|\phi_k^A\rangle\}$, $\{|\phi^B_k\rangle\}$ are two orthonormal
subsets of $A$ and $B$, respectively \cite{schmidt}.  This
decomposition takes the name of \emph{Schmidt
  decomposition\footnote{More generally, the Schmidt decomposition is
    well defined for pure states belonging to general Hilbert spaces
    $X$. }}.
\end{definition}
\begin{remark}\label{rmktrace}
Consider the decomposition $\CC^{2^n}=A\otimes B$ and let
$\rho_{A,\psi}$ be the density operator of the state $|\psi\rangle$ on
the subsystem $A$.  Then the set of the positive eigenvalues of
$\rho_{A,\psi}$ coincides with the set $\{c_k^2\mid c_k>0\}$ of
positive squared coefficients of Schmidt decomposition of the state
$|\psi\rangle$ with respect to the decomposition $\CC^{2^n}=A\otimes
B$ -- see for instance \cite{gentle}. As a consequence,
$Tr[\rho_{A,\psi}]=\sum_{k=1}^K c_k^2=1$ and
$Tr[\rho^2_{A,\psi}]=\sum_{k=1}^K c_k^4$.
\end{remark}

\begin{proposition}\label{pmax}
If $|\FF(|\psi\rangle)|=1$ then $|\psi\rangle$ is maximally entangled
with respect to MW measure.
\end{proposition}
\begin{proof}
First of all we notice that
 \begin{enumerate}[(a)]
  \item If $|\phi\rangle\in \CC^2$ then $\langle \phi| \sigma_y |\bar 
\phi\rangle=0$;
  
  \item If $\{|\phi_1\rangle,|\phi_2\rangle\}$ is an orthonormal base of $\CC^2$ 
then 
  $|\langle \phi_1| \sigma_y |\bar \phi_2\rangle|=1$
  and
  $$\langle \phi_1| \sigma_y |\bar \phi_2\rangle=- \langle \phi_2| \sigma_y 
|\bar \phi_1\rangle$$
  
  \item For all $|\xi_1\rangle,|\xi_2\rangle\in \CC^{2^{n-1}}$ one has 
  $$|\langle \xi_1|  I_{2^{n-2}}\otimes\sigma_y| \bar \xi_2\rangle|\leq 1;$$
  and
  $$\langle \xi_1|  I_{2^{n-2}}\otimes\sigma_y \bar \xi_2\rangle= -\langle 
\xi_2|  I_{2^{n-2}}\otimes\sigma_y |\bar \xi_1\rangle. $$
   \end{enumerate}
Also remark that the Schmidt decomposition of $|\psi\rangle$ with
respect the decomposition that singles out a generic qubit of the
system reads:
$$|\psi\rangle=\sum_{k=1}^2 c_k |\phi_k\rangle\otimes |\xi_k\rangle$$
for some $c_1,c_2\geq0$ such that $c_1^2+c_2^2=1$, some orthonormal base 
$\{|\phi_1\rangle,|\phi_2\rangle\}$ of $\CC^2$ and some orthonormal subset 
$\{|\xi_1\rangle,|\xi_2\rangle\}$ of $\CC^{2^{n-1}}$.

In view of (a)-(c), we then have
\begin{align*}
|\FF(|\psi\rangle)|&=|\sum_{k,h=1}^2  c_k c_h \langle \phi_k| \sigma_y  
|\phi_h\rangle\langle \xi_k|  I_{2^{n-2}}\otimes\sigma_y |\bar \xi_h\rangle |\\
       &=|2 c_1c_2 \langle \xi_1|  I_{2^{n-2}}\otimes\sigma_y| \bar 
\xi_2\rangle|\leq |2 c_1c_2|
\end{align*}

On the other hand $|2 c_1c_2|\leq 1$ for all $c_1,c_2\in\RR$ such that
$c_1^2+c_2^2=1$, and the maximum $|2 c_1c_2|=1$ is attained at the
points satisfying $c_1^2=c_2^2=1/2$.  Therefore we may conclude that
if $|\FF(|\psi\rangle)|=1$ then $c_1^2=c_2^2=1/2$.  Since this
argument holds for any qubit, we have that
$$Tr[\rho_{j,\psi}^2]=c_1^4+c_2^4=\frac{1}{2} \quad \text{ for all 
$j=1,\dots,n$}$$
see also Remark \ref{rmktrace}.
Consequently,  
$$Q(|\psi\rangle)=2\left(1-\frac{1}{n}\sum_{j=1}^{n} 
Tr[\rho_{j,\psi}^2]\right)=1.$$
\end{proof}

Above results relate the value of $|\mathcal F(\ket \psi)|$ to a
measure of entanglement of the state $\ket \psi$. In particular if
$|\mathcal F(\ket \psi)|$ is minimal, i.e., $|\mathcal F(\ket
\psi)|=0$, then $\ket \psi$ is not entangled while if $|\mathcal
F(\ket \psi)|$ is maximal, i.e., $|\mathcal F(\ket \psi)|=1$ then
$\ket \psi$ is maximally entangled. However the condition $|\mathcal
F(\ket \psi)|=0$ (respectively $|\mathcal F(\ket \psi)|=1$) is a necessary
 but not  sufficient condition to have $\ket \psi$ unentangled
(resp. maximally entangled).  Indeed, consider the
\emph{Greenberger-Horne-Zeilinger} state
$$|GHZ_n\rangle:=\frac{1}{\sqrt{2}}(|0_n\rangle+|1_n\rangle).$$
 
For all $n\geq 2$, the state $|GHZ_n\rangle$ is globally entangled
state and yet, for $n\geq 3$, $\FF(|GHZ_n\rangle)=0$: this implies
that, in general, the inverse implication of Proposition \ref{pcrit}
(that is, $\FF(\ket \psi)=0$ implies $\ket \psi$ is unentangled) is
not true.  Furthermore, for all $n\geq 2$, the state $|GHZ_n\rangle$
is maximally entangled with respect to MW measure and $\FF(\ket
\psi)\not=1$, thus also the inverse implication of Proposition
\ref{pmax} (that is, $\FF(\ket \psi)=1$ implies $\ket \psi$ is
maximally entangled) in general is not true.

\section{$n$-qubit entanglement algorithm}\label{s3}
In this section we introduce a generalisation of the \textsc{cnot}
gate and we show that the resulting Bell state are fully entangled.

\begin{figure}\label{circuito}
\begin{center}
\definecolor{white}{RGB}{255,255,255}
\definecolor{cffcc00}{RGB}{230,230,230}
\begin{tikzpicture}[y=0.80pt, x=0.8pt,yscale=-1, inner sep=0pt, outer sep=0pt,
draw=black,fill=black,line join=miter,line cap=rect,miter limit=10.00,line 
width=0.800pt]
  \begin{scope}[shift={(-333.0,-291.0)},draw=white,fill=white]
    \path[fill,rounded corners=0.0000cm] (333.0000,291.0000) rectangle
      (676.0000,471.0000);
  \end{scope}
  
\begin{scope}[cm={{1.0,0.0,0.0,1.0,(-333.0,-291.0)}},draw=cffcc00,fill=cffcc00]
    \path[fill,rounded corners=0.0000cm] (458.0000,306.0000) rectangle
      (488.0000,336.0000);
  \end{scope}
  \begin{scope}[cm={{1.0,0.0,0.0,1.0,(-333.0,-291.0)}},line cap=butt,miter 
limit=1.45]
    \path[fill] (464.5908,325.5352) node[above right] (text4235) {$H_2$};
    \path[draw,rounded corners=0.0000cm] (458.0000,306.0000) rectangle
      (488.0000,336.0000);
    \path[fill] (380.2402,325.5352) node[above right] (text4239) {$\ket0$};
    \path[fill] (467.5986,385.5352) node[above right] (text4241) {$\otimes$};
    \path[fill] (611.5889,385.5352) node[above right] (text4243) {$\beta$};
  \end{scope}
  
\begin{scope}[cm={{1.0,0.0,0.0,1.0,(-333.0,-291.0)}},draw=cffcc00,fill=cffcc00]
    \path[fill,rounded corners=0.0000cm] (539.0000,366.0000) rectangle
      (569.0000,396.0000);
  \end{scope}
  \begin{scope}[cm={{1.0,0.0,0.0,1.0,(-333.0,-291.0)}},line cap=butt,miter 
limit=1.45]
    \path[fill] (540.9609,383.0352) node[above right] (text4251) {\tiny CNOT};
    \path[draw,rounded corners=0.0000cm] (539.0000,366.0000) rectangle
      (569.0000,396.0000);
    \path[fill] (380.2402,445.5352) node[above right] (text4255) {$\ket1$};
    \path[draw] (407.0000,321.0000) -- (450.0000,321.0000);
    \path[fill] (458.0000,321.0000) -- (446.0000,316.0000) -- 
(449.0000,321.0000) --
      (446.0000,326.0000) -- cycle;
    \path[draw] (473.0000,336.0000) -- (473.0000,371.5000);
    \path[fill] (473.0000,379.5000-5) -- (478.0000,367.5000-5) -- 
(473.0000,370.5000-5) --
      (468.0000,367.5000-5) -- cycle;
    \path[draw] (480.9960,381.0000) -- (531.0000,381.0000);
    \path[fill] (539.0000,381.0000) -- (527.0000,376.0000) -- 
(530.0000,381.0000) --
      (527.0000,386.0000) -- cycle;
    \path[draw] (569.0000,381.0000) -- (602.0000,381.0000);
    \path[fill] (610.0000,381.0000) -- (598.0000,376.0000) -- 
(601.0000,381.0000) --
      (598.0000,386.0000) -- cycle;
    \path[draw] (407.0000,441.0000) -- (473.0000,441.0000) -- 
(473.0000,390.5000);
    \path[fill] (473.0000,387.5000) -- (468.0000,394.5000+5) -- 
(473.0000,391.5000+5) --
      (478.0000,394.5000+5) -- cycle;
  \end{scope}
\end{tikzpicture}
\end{center}
\caption{Bell Circuit: entanglement of two elements of the canonical basis 
$\ket0$ and $\ket1$}
\end{figure}

To this end we adopt the following notations:
\begin{notation}
We use $H_2:= \frac{1}{\sqrt{2}} \begin{pmatrix}1 & 1 \\ 1 & -1
\end{pmatrix}$ is the Hadamard matrix and $$H_{2^n}:=\underbrace{H_2\otimes 
\dots \otimes H_2}_{n\text{ times}}$$ is its $2^n$-dimensional
generalisation, i.e., the $2^n$-dimensional Walsh matrix.  We use the
symbols $\sigma_x ,\sigma_y$ and $\sigma_z$ to denote Pauli's matrices
 $$\sigma_x =\begin{pmatrix}
             0&1\\
             1&0
            \end{pmatrix},\qquad
     \sigma_y=\begin{pmatrix}
             0&-i\\
             i&0
            \end{pmatrix}, \qquad
       \sigma_z=\begin{pmatrix}
             1&0\\
             0&-1
            \end{pmatrix}.$$
We finally consider the orthogonal projectors
$$L := \left(
\begin{array}{cc}
1 & 0  \\
0 & 0 
\end{array}
\right), \qquad
R := \left(
\begin{array}{cc}
0 & 0  \\
0 & 1 
\end{array}
\right).$$
\end{notation}
In view of above notation, we remark that the \textsc{cnot} gate
satisfies the equality
\begin{equation*}\label{c2}
 \text{\textsc{cnot}}:=\begin{pmatrix}
                1&0&0&0\\
                0&1&0&0\\
                0&0&0&1\\
                0&0&1&0\\
                \end{pmatrix}=L\otimes I_2+R\otimes\sigma_x 
\end{equation*}
while the columns of the matrix
\begin{equation*}\label{b2}
 B_4:=\frac{1}{\sqrt{2}}\begin{pmatrix}
                1&0&1&0\\
                0&1&0&1\\
                0&1&0&-1\\
                1&0&-1&0\\
                \end{pmatrix}=\text{\textsc{cnot}} (H_2\otimes I_2)
\end{equation*}
are the coordinate vectors of the Bell states in the standard base. We extend 
the above definitions of \textsc{cnot} and of $B_2$ to an arbitrary number of 
qubits as follows
\begin{figure}
\begin{center}
\begin{tabular}{ccc}
\includegraphics[width=0.3\linewidth]{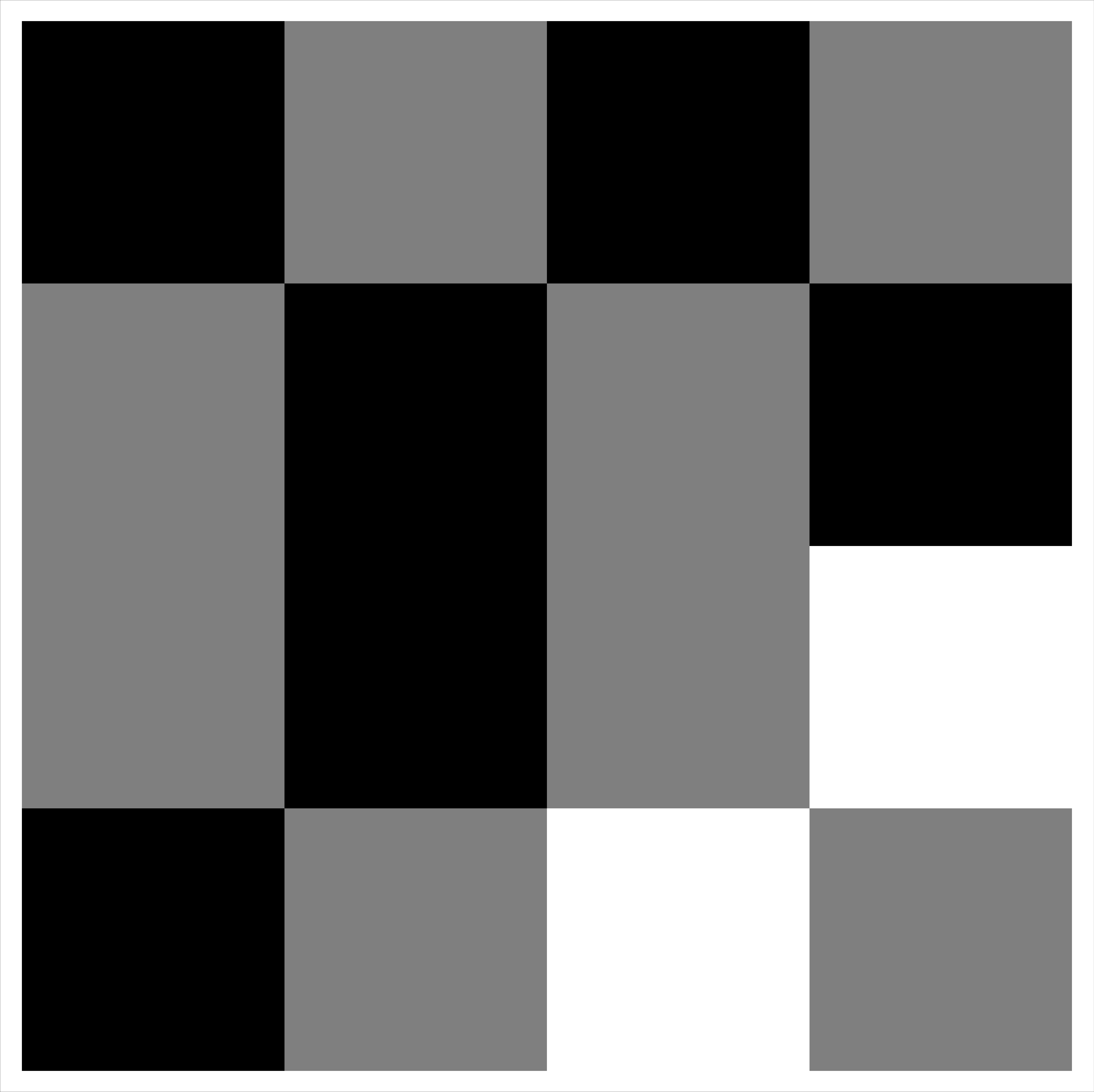} &
 \includegraphics[width=0.3\linewidth]{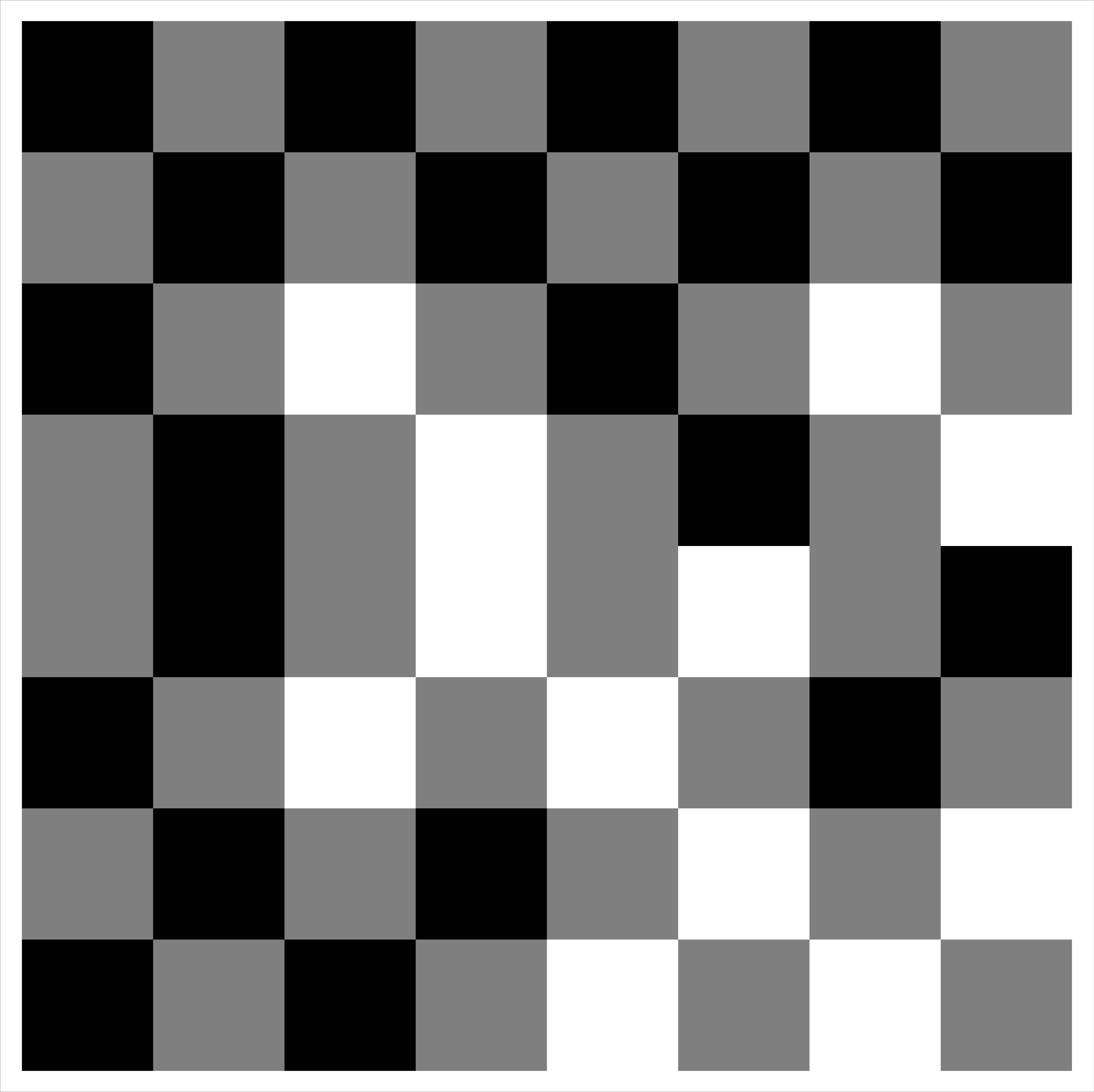} &
 \includegraphics[width=0.3\linewidth]{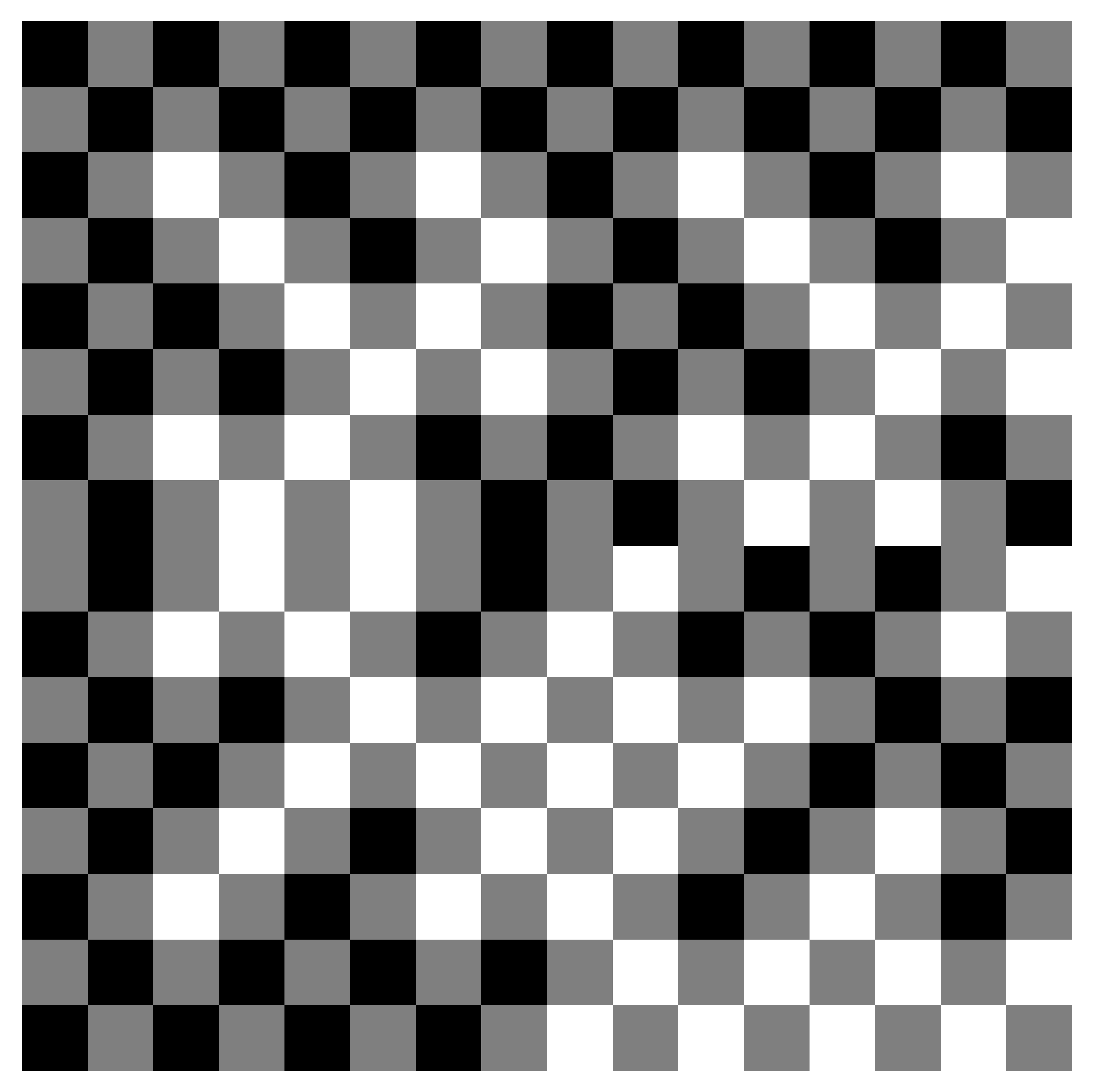} \\
 \includegraphics[width=0.3\linewidth]{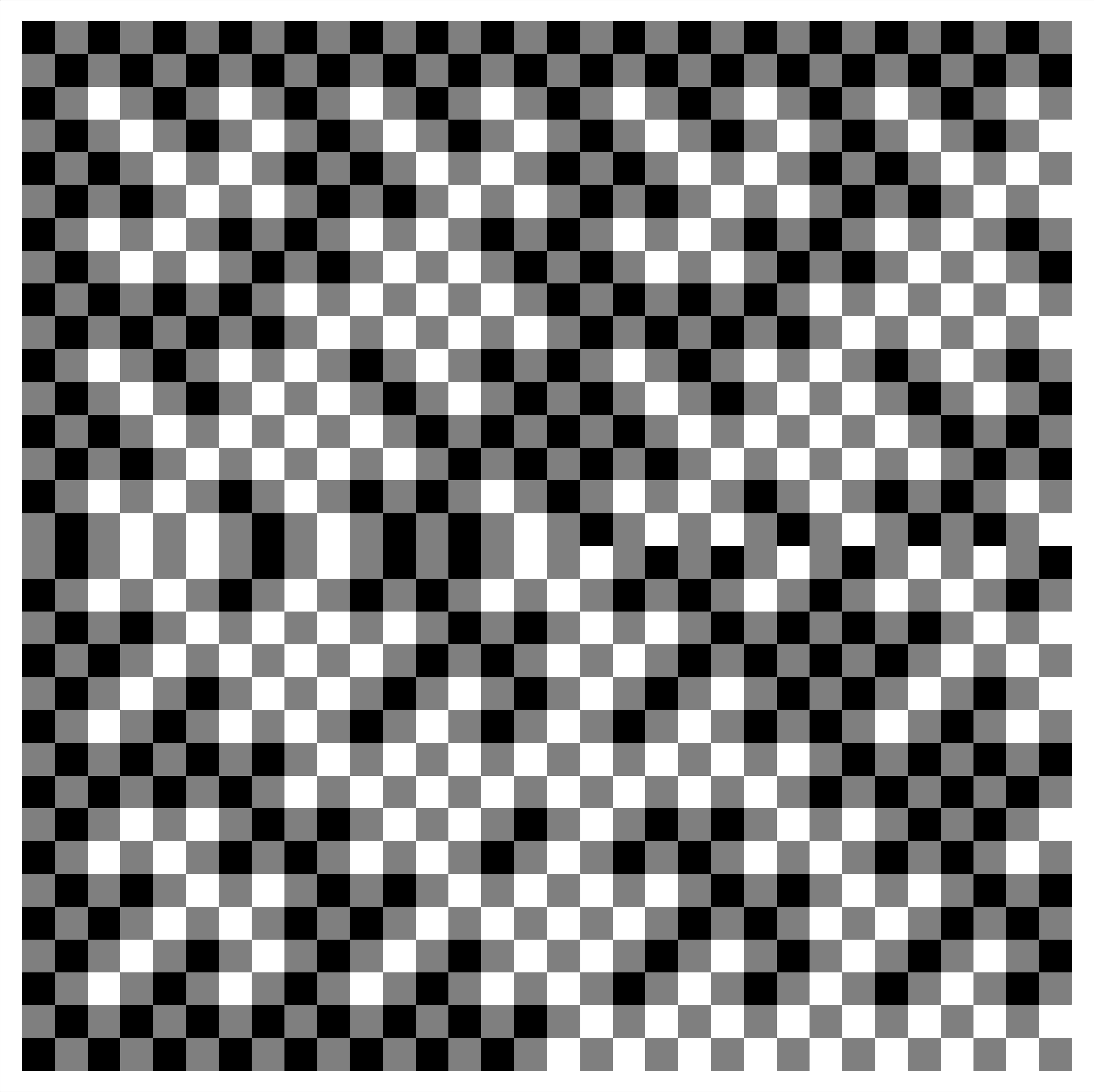} &
 \includegraphics[width=0.3\linewidth]{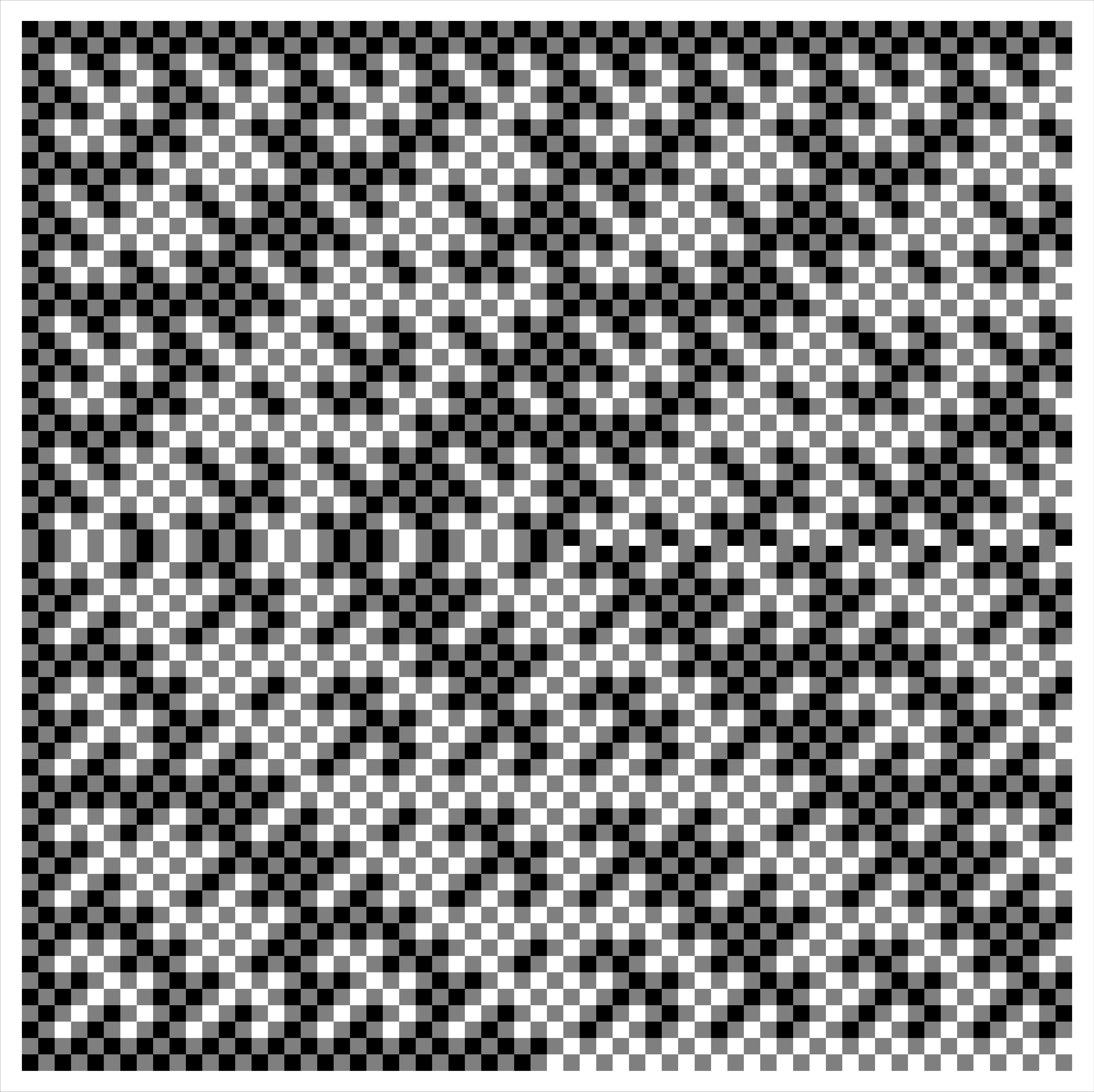} &
 \includegraphics[width=0.3\linewidth]{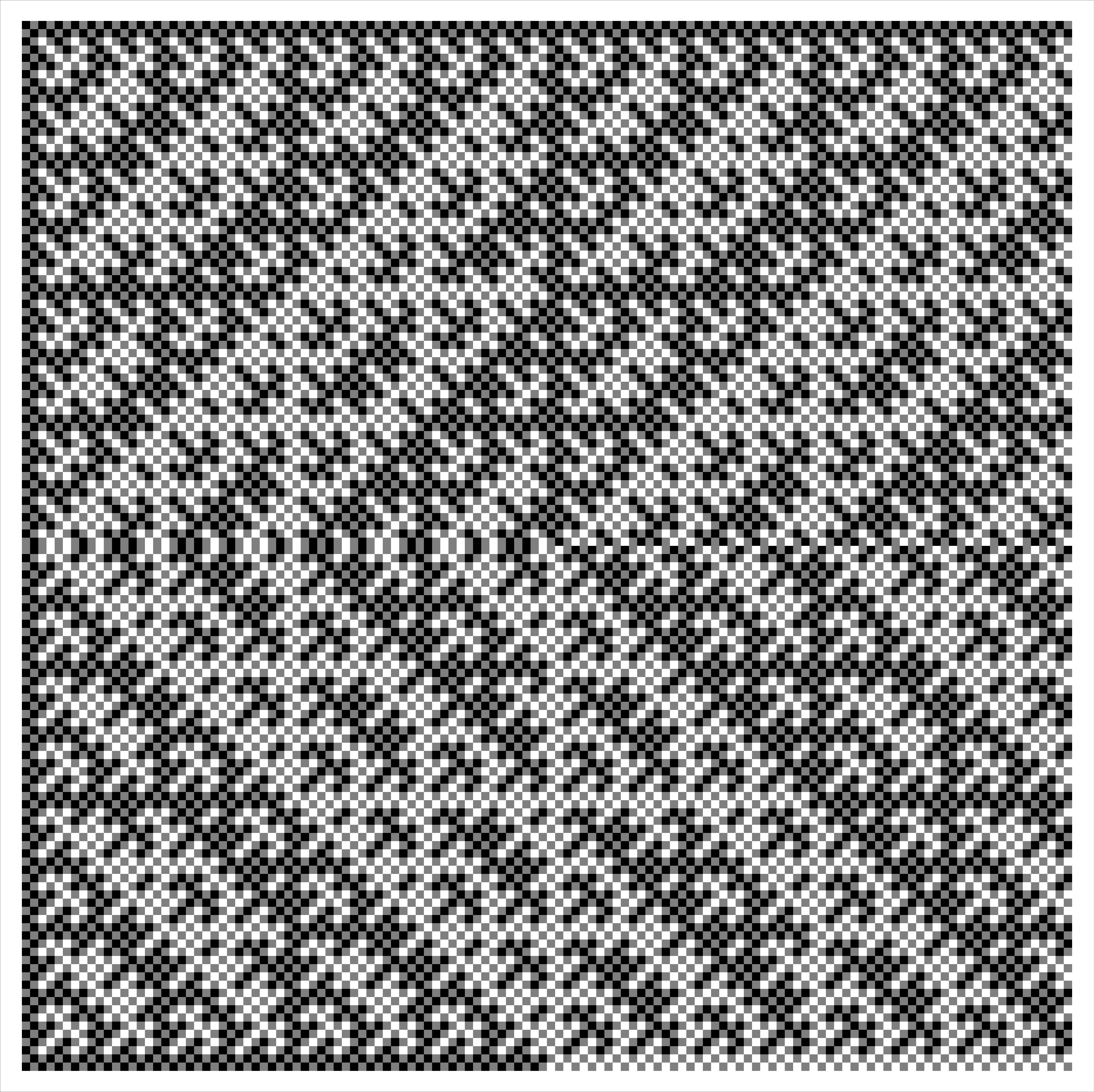} \\ 
\end{tabular}
\end{center}
\caption{We show in this picture Bell matrices $B_{2^n}$ for
  $n=2,\dots , 7$.  Entries $0$ are shown in grey color, entries $+1$
  by black color and entries $-1$ by white color.}\label{hadamards}
\end{figure}
\begin{definition}\label{defcnot}
For $n\geq 2$ we set
\begin{equation}\label{cn}
\text{\textsc{cnot}}_{2^n} :=  L\otimes I_{2^{n-1}}  + R\otimes  
\underbrace{\sigma_x \otimes \dots \otimes \sigma_x}_{n-1\text{ times}}
 \end{equation}
and \begin{equation}\label{bn}
 B_{2^n} := \text{\textsc{cnot}}_{2^n}(H_{2^{n-1}}\otimes I_2).
\end{equation}
We define \emph{$2^n$-dimensional Bell state} any state
$$|b_k\rangle:=B_{2^n}|k\rangle$$
where $ k=0,\dots,2^{n}-1$ and $|k\rangle$ is the $k$-th element of the standard 
base of $\CC^{2^n}$. 
\end{definition}
In what follows we show that the $2^n$-dimensional Bell states are
maximally entangled with respect to MW measure. We introduce the
matrix
\begin{equation}\label{Ldef}
L_{2^n}:=B_{2^n}^\dagger  M_{2^n}B_{2^n}, 
\end{equation}
whose relevance in our investigation is motivated by the following 
\begin{lemma}\label{mot}
If $|\langle \phi |L_{2^n}| \bar \phi\rangle|=1$ and if $|\psi\rangle= B_{2^n} 
|\phi\rangle$ then $|\psi\rangle$ is maximally entangled with respect to the MW 
measure. 

In particular, if $|\langle k | L_{2^n}|\bar k\rangle |=1$, where
$|k\rangle$ is the $k$-th element of the standard base, then the
$k$-th Bell state is maximally entangled with respect to the MW
measure.
\end{lemma}
\begin{proof}
By the definition of $L_{2^n}$ and by the assumption $|\psi\rangle=
B_{2^n} |\phi\rangle$ one has
\begin{align*}
|\langle \phi |L_{2^n}| \bar \phi\rangle|&=|\langle \phi | B_{2^n}^\dagger  
M_{2^n}B_{2^n} | \bar \phi\rangle|\\
&=|\langle B_{2^n} \phi | M_{2^n} | \overline{(B_{2^n}|\phi)}\rangle|=|\langle 
\psi | M_{2^n} | \bar \psi)\rangle|=|\FF(|\psi\rangle)|.
 \end{align*}
The first part of the claim hence follows by Proposition \ref{pmax}.

The second part of the claim readily follows by applying above
reasoning to $|\phi\rangle=|k\rangle$ and by the definition of
$2^n$-dimensional Bell state.
\end{proof}
\begin{remark}
There exist states $\phi$ which not satisfy $|\langle \phi |L_{2^n}|
\bar \phi\rangle|=1$ and such that $B_{2^n} |\phi\rangle$ is maximally
entangled, an example of this phenomenon is given by the state
$\phi=B^{-1}_{2^n}\ket{GHZ_n}$.
\end{remark}

Next result gives a closed formula for $L_{2^n}$ and relates its
diagonal elements to the \emph{Thue-Morse sequence}, that is the
binary sequence $(\tau_i)$ defined by the recursive relation
  \begin{align*}
   &\tau_1:=0\\
   &\tau_{2n}:=1-\tau_n\\
   &\tau_{2n-1}:=\tau_n  \end{align*}
for all positive integers $n$. We notice that for all $n\geq 1$
\begin{equation}\label{thuechar}
\tau_{2^n+i}=1-\tau_{i}\quad \text{for all $i=1,\dots,2^n$}.
\end{equation}
\begin{remark}
Equality \eqref{thuechar} characterises the Thue-Morse sequence via
bitwise negation, indeed it states that every initial block of length
$2^n$, i.e, $\tau_1,\dots,\tau_{2^n}$, is followed by a block of equal
length that is its bitwise negation, i.e.,
$\tau_{2^n+1}=1-\tau_{1},\dots,\tau_{2^{n+1}}=1-\tau_{2^n}$.  This can
be proved by an inductive argument, indeed the case $n=1$ follows by a
direct computation and, assuming \eqref{thuechar} as inductive
hypothesis, one readily gets the inductive step
\begin{equation*}
\tau_{2^{n+1}+i}=\begin{cases}
\tau_{2(2^{n}+i/2)}=1-\tau_{2^{n}+i/2}=1-\tau_{i/2}=1-\tau_i &\quad\text{if $i$ 
is even;}\\
\tau_{2(2^{n}+i/2)}=\tau_{2^{n}+(i+1)/2}=1-\tau_{(i+1)/2}=1-\tau_i 
&\quad\text{if $i$ is odd.}
                 \end{cases}
\end{equation*}
\end{remark}

\begin{lemma}\label{lL}For all $n\geq 2$
 \begin{equation}
 \label{LT1}
L_{2^n}=- \underbrace{\sigma_z \otimes \dots \otimes \sigma_z}_{n\text{ times}}. 
 \end{equation}
Moreover $L_{2^n}$ is a diagonal matrix whose diagonal elements
$L_{2^n,i}$, ${i=1,\dots, 2^n}$, satisfy
\begin{equation}\label{thue}
L_{2^n,i}=2 \tau_i-1, \qquad \text{for all $n=1,\dots,2^n$},
\end{equation}
where $(\tau_i)$ is the Thue-Morse sequence.
\end{lemma}
\begin{proof}
In order to prove (\ref{LT1}), we recall the definition of $L_{2^n}$ in Equation 
\eqref{Ldef}
 \begin{equation}\label{n2}
  \begin{split}
  L_{2^n}=&B^\dagger_{2^n} M_{2^n} B_{2^n}\\ =&(LH_2)^\dagger \sigma_y
  RH_2 \otimes \underbrace{(H_2^\dagger \sigma_x
    H_2)\otimes\dots\otimes (H_2^\dagger \sigma_x H_2)}_{n-1\text{
      times}}\otimes \sigma_y\sigma_x +\\ &+(RH_2)^\dagger \sigma_y
  LH_2\otimes\underbrace{((\sigma_x H_2)^\dagger
    H_2)\otimes\dots\otimes ((\sigma_x H_2)^\dagger H_2)}_{n-1 \text{
      times}}\otimes \sigma_x ^\dagger \sigma_y
  \end{split}  
\end{equation}
the second equality is obtained by applying Definition \ref{defcnot},
Equations \eqref{cn} and \eqref{bn} where $B_{2^n}$ is given in terms
of $\text{\textsc{cnot}}_{2^n}$, namely
$$B_{2^n}= (LH_2)\otimes\underbrace{H_2\otimes \dots \otimes
  H_2}_{(n-1)\text{ times}} + (RH_2)\otimes \underbrace{\sigma_x H_2
  \otimes \dots \otimes \sigma_x H_2}_{(n-1)\text{ times}}$$ 
and by applying $L^\dagger \sigma_y L=R^\dagger \sigma_y R=0$.  By a
direct computation
$$(RH_2)^\dagger \sigma_y LH_2= ((LH_2)^\dagger \sigma_y RH_2)^\dagger 
=-\frac{i}{2}\begin{pmatrix}
                                    -1&1\\
                                    -1&1\\
\end{pmatrix} \quad\text{and}\quad \sigma_x ^\dagger \sigma_y=(\sigma_y\sigma_x 
)^\dagger =i\sigma_z.$$
By plugging above relations in (\ref{n2}) we obtain the first part of
the claim, indeed
 \begin{align*}
  L_{2^n}=&\frac{i}{2}\begin{pmatrix}
                       -1&1\\
                       -1&1\\
 \end{pmatrix}\otimes\underbrace{\sigma_z\otimes\dots\otimes\sigma_z}_{n-2 \text{ 
times}} \otimes (-i \sigma_z )-
\frac{i}{2}\begin{pmatrix}
                       -1&-1\\
                       1&1\\
            \end{pmatrix}\otimes\underbrace{\sigma_z\otimes\dots\otimes\sigma_z}_{n-2 \text{ 
times}}\otimes i \sigma_z \\
                      =&\frac{1}{2}\begin{pmatrix}
                       -1&1\\
                       -1&1\\
                      \end{pmatrix}\otimes\underbrace{\sigma_z\otimes\dots\otimes\sigma_z}_{n-1 \text{ 
times}}+
\frac{1}{2}\begin{pmatrix}
                       -1&-1\\
                       1&1\\
                   \end{pmatrix}\otimes\underbrace{\sigma_z\otimes\dots\otimes\sigma_z}_{n-1 \text{ 
times}}\\
                      =&-\underbrace{\sigma_z\otimes\dots\otimes\sigma_z}_{n 
\text{ times}}.
                      \end{align*}                     
Now, above equality implies
\begin{equation}\label{LTT}
 L_{2^n}=\sigma_z \otimes L_{2^{n-1}}
\end{equation}
and, by an inductive argument, that $L_{2^n}$ is a diagonal matrix.

Finally we prove \eqref{thue} by induction on $n$.  The base of
induction, i.e. the case $n=1$, readily follows by $L_2=\sigma_z$ and
by the definition of $\tau_1$ and of $\tau_2$.  Now we prove the
inductive step, i.e., we assume \eqref{thue} as inductive hypothesis
and we show
\begin{equation}\label{thue2}
  L_{2^{n+1},i}=2\tau_i-1, \qquad \text{for all $i=1,\dots,2^{n+1}$},
 \end{equation}
 By \eqref{LTT} we have $L_{2^{n+1}}=\sigma_z \otimes L_{2^{n}}$ and, 
consequently,
$$L_{2^{n+1},i}=\begin{cases}
             L_{2^{n},i}\quad &\text{if } i\leq 2^{n}\\
             -L_{2^{n},i-2^{n}}\quad &\text{otherwise}.
            \end{cases}.$$
This, together with \eqref{thuechar}, implies \eqref{thue2}, indeed we
have
\begin{align*}
&L_{2^{n+1},i}= L_{2^{n},i}=2 \tau_i-1\\
&L_{2^{n+1},2^n+i}= -L_{2^{n},i}=1-2 \tau_i=2 \tau_{2^n+i}-1
 \end{align*}
for all $i=1,\dots,2^n$ and this completes the proof.
\end{proof}

\begin{theorem}\label{thm}
The $2^n$-dimensional Bell states are maximally entangled with respect
to MW measure.
\end{theorem}
\begin{proof}
By Lemma \ref{lL}, $L_{2^n}$ is a diagonal matrix with $1$ or $-1$ as
diagonal elements then $ |\langle k|L_{2^n}|\bar k\rangle|=1$ for all
$k=0,\dots,2^n-1$ and this, together with Lemma \ref{mot}, implies the
claim.
\end{proof}

\subsection{Some remarks on an entanglement criterion}\label{s31} 
Lemma~\ref{mot} provides a maximal entanglement criterion that can be
rephrased as follows ``If $|\langle \phi | L_{2^n} |\bar\phi \rangle
|=1$ then $B_{2^n}| \phi\rangle$ is maximally entangled''.  Then one
may ask how is made the space of states satisfying this condition.
Lemma~\ref{lL} provides some answers to this question.  Indeed we
already used in the proof of Theorem~\ref{thm} the fact that $|\langle
k| L_{2^n}| k\rangle |=1$, if $|k\rangle$ is an element of the
canonical base.  Next result investigates this property in the larger
class of states whose coordinates in the standard base are real
valued.
\begin{proposition}\label{evil}
Let $(\tau_i)$ be the Thue-Morse sequence and let $o_i$ and $e_i$ be
the index sequences such that $\tau_{o_i}=1$ and $\tau_{e_i}=0$ for
all $i\in\NN$.  Then for all $x\in \RR^{2^n}$ with $|x|=1$, one has
$|x\tsp L_{2^n} x|=1$ if and only if either $x_{e_i}=0$ for all
$i=1,\dots,2^{n-1}$ or $x_{o_i}=0$ for all $i=1,\dots,2^{n-1}$.
\end{proposition}

\begin{proof}
Let $x=(x_1,\dots,x_{2^n})\in \RR^{2^n}$ with $|x|=1$. Since $x$ is
real valued then $|x_i|^2=x_i^2$ for all $i=1,\ldots,2^n$ and
$\sum_{i=1}^{2^n} x^2_i=|x|^2$.  On the other hand $|x\tsp L_{2^n}
x|=|\sum_{i=1}^{2^n} \tau_i x^2_i|=1$ if and only if either
$\sum_{i=1}^{2^n} L_{2^n,i} x^2_i=1$ or $\sum_{i=1}^{2^n} L_{2^n,i}
x^2_i=-1$, where $L_{2^n,i}$ is the $i$-th diagonal element of
$L_{2^n}$.  Since $|x|=1$, the former case is equivalent to
$\sum_{i=1}^{2^n} L_{2^n,i} x^2_i=\sum_{i=1}^{2^n} x^2_i$ and, this,
together with the equality $L_{2^n,i}=2\tau_i-1$ proved in Lemma
\ref{lL}, implies
$$0=\sum_{i=1}^{2^n} (L_{2^n,i}-1) x^2_i=\sum_{i=1}^{2^n} (2\tau_i-2) 
x^2_i=-2\sum_{i=1}^{2^{n-1}}  x^2_{e_i}.$$

Above equality holds if and only if $x_{e_i}=0$ for all
$i=1,\dots,2^{n-1}$.  It follows by a similar argument that
$\sum_{i=1}^{2^n} \tau_i x^2_i=-1$ is equivalent to $x_{o_i}=0$ for
all $i=1,\dots,2^{n-1}$ and this completes the proof.
\end{proof}

\begin{remark}
The index sequences $e_i$ and $o_i$ defined in above Proposition are
called Conway's odious and evil numbers.
\end{remark}

\section{Conclusions}
We proposed a family of unitary transformations generalising the
\textsc{cnot} gate to an arbitrary number of qubits. We showed that a
circuit composed by Walsh matrix and our general \textsc{cnot} gate
yields a maximally entangled (with respect to MW measure) set of
states, that we called \emph{generalised Bell states}.  In order to
prove the validity of the method, we developed ad hoc entanglement
criteria based on the definition of a suitable antilinear operator.
The paper also contains a preliminary theoretical investigation of
such operator, which turned out to be related with the celebrated
Thue-Morse sequence.

Results in the present paper open the way to further investigations in
several directions. For instance, it could be interesting to extend
the method to general controlled unitary operations.  Also, a deeper
investigation of antilinear operators with zero expectation value on
product states could represent a step towards an algebraic
characterisation of the states with maximal MW measure. Finally it
could be interesting to better understand the intriguing relation
between states with maximal MW measure and the Thue-Morse sequence.

\bibliographystyle{plain} 
\bibliography{quantum}
\end{document}